\documentclass[12pt]{article}
\usepackage{fullpage,amsfonts,amssymb,amsmath}
\usepackage{fullpage,amsfonts,amssymb,amsmath,amsthm,setspace}

\newcounter {ctr1}

\newtheorem{mythm}[ctr1]{Theorem}
\newtheorem{mylem}[ctr1]{Lemma}

\DeclareMathOperator{\newmod}{mod}

\DeclareMathOperator{\dom}{dom}

\renewcommand{\mathbf}[1]{{\pmb #1}}

\newcommand{\bE}{\mathbb{E}}
\newcommand{\bZ}{\mathbb{Z}}

\def\blackslug{\hbox{\hskip 1pt \vrule width 8pt height 8pt depth 0pt
\hskip 1pt}}

\def\bqed{\quad\blackslug\lower 8.5pt\null\par}

\def\wqed
{\quad\raisebox{-.3ex}{\Large$\Box$}\lower 8.5pt\null\par}
\long\gdef\boxit#1{\begingroup\vbox{\hrule\hbox{\vrule\kern3pt
      \vbox{\kern3pt#1\kern3pt}\kern3pt\vrule}\hrule}\endgroup}

\newlength{\saveparindent}
\newlength{\saveparskip}
\newcounter{ctr}

\begin{document}

\title{Block Sensitivity of Minterm-Transitive Functions}

\author{
Andrew Drucker 
\\ MIT}

\maketitle

\begin{abstract}
Boolean functions with symmetry properties are interesting from a complexity theory perspective; extensive research has shown that these functions, if nonconstant, must have high `complexity' according to various measures.

In recent work of this type, Sun gave bounds on the block sensitivity of nonconstant Boolean functions invariant under a transitive permutation group.  Sun showed that all such functions satisfy $bs(f) = \Omega(N^{1/3})$, and that there exists such a function for which $bs(f) = O(N^{3/7}\ln N)$.  His example function belongs to a subclass of transitively invariant functions called the \emph{minterm-transitive} functions (defined in earlier work by Chakraborty).  

We extend these results in two ways.  First, we show that nonconstant minterm-transitive functions satisfy $bs(f) = \Omega(N^{3/7})$.  Thus Sun's example function has nearly minimal block sensitivity for this subclass.  Second, we give an improved example: a minterm-transitive function for which $bs(f) = O(N^{3/7}\ln^{1/7}N)$.

\end{abstract}

\section{Introduction}

Boolean functions, like other objects in mathematics, can be classified according to the symmetries they possess.  A natural notion of symmetry arises when we consider permutations of the input variables.  Given a function $f: \{0,1\}^N \rightarrow \{0, 1\}$ and a permutation $\sigma$ on $\{1, \ldots, N\}$, we say that $f$ is invariant under $\sigma$ if permuting the input variables according to $\sigma$ never affects the value of $f$.  For every function $f$ it is easily seen that the set of permutations under which $f$ is invariant forms a group, the \emph{invariance group} of $f$.

One class of `high-symmetry' functions are those whose invariance group is \emph{transitive}: a permutation group $\Gamma$ is transitive if for each $i, j \in [N]$ there is a $\pi \in \Gamma$ such that $\pi(i)= j$.  Transitively invariant Boolean functions (also called `weakly symmetric functions') are a natural, important class which includes graph properties and symmetric functions.  They are of particular interest in computational complexity: several decades of research have shown that certain classes of (nonconstant) transitively invariant Boolean functions have high `complexity' in several senses.  
For example, symmetric functions on $N$ inputs have randomized query complexity $\Omega(N)$, quantum query complexity $\Omega(\sqrt{N})$ \cite{BBC+}, and sensitivity $\Omega(N)$, while graph properties on $n$-vertex graphs have deterministic query complexity $\Omega(n)$, quantum query complexity $\Omega(n^{1/2})$~\cite{SYZ}, and sensitivity $\Omega(n)$~\cite{Tur}.  In each case the lower bound obtained is tight (except for a log-factor uncertainty in the case of~\cite{SYZ}).

For general transitively invariant functions, the deterministic and quantum query complexities have also been pinpointed fairly precisely \cite{SYZ}; however, the sensitivity and block sensitivity of these functions are less well understood.  In particular, it is open whether such functions have sensitivity $s(f) = N^{\Omega(1)}$.  A version of this question was first asked by Turan in 1984~\cite{Tur}, who gave an affirmative answer for the case of graph properties.

Partial progress on Turan's question was made by Chakraborty, who in~\cite{Cha} defined a special class of transitively invariant functions called \textit{minterm-transitive functions} (see Section~\ref{ss2} for the definition).  These functions are of interest because although they are of restricted form, they place no restriction on the type of transitive invariance group associated with the Boolean function (in contrast to graph properties and symmetric functions).  Chakraborty showed that for such functions $s(f) = \Omega(N^{1/3})$, and he also constructed an example for which this bound is tight.  This is the lowest sensitivity known for any transitively invariant function.  

In subsequent work, Sun~\cite{Sun} showed that for general transitively invariant functions, the block sensitivity $bs(f)$ satisfies $bs(f) = \Omega(N^{1/3})$.  Sun also gave an example of a transitively invariant (in fact minterm-transitive) function for which $bs(f) = O(N^{3/7}\ln N)$.

In this paper, we extend Sun's results in two directions.  First, we show that for minterm-transitive functions, $bs(f) = \Omega(N^{3/7})$.  While this does not close the gap in our knowledge for general transitively invariant functions, it in a sense explains why Sun's upper bound took the form it did.  To prove this result (in Section~\ref{lboundsec}), we build on Sun's approach (which is also related to ideas in~\cite{Cha}, \cite{RV}) of selecting random permutations from the invariance group for $f$ to find disjoint sensitive blocks.  In a novel step, we use the `deletion method' of probabilistic combinatorics \cite{AS} to create a large collection of sensitive blocks with `low overlap'; we then apply a simple method for passing from an input with many, low-overlap sensitive blocks to an input with many disjoint sensitive blocks.

Next, we improve Sun's upper-bound example, by giving (in Section~\ref{uboundsec}) a family of minterm-transitive functions for which $bs(f) = O(N^{3/7}\ln^{1/7} N)$.  Our basic approach is the same as Sun's \cite{Sun}, but we improve part of the construction, using a powerful inequality from probability theory due to Janson and Suen \cite{Jan}.  We introduce this inequality in Section~\ref{probineqsec}.

\section{Preliminaries\label{prelim}}

For convenience, in this paper we will always regard an $N$-bit string as having coordinates indexed by $\bZ_N = \{0,1 , \ldots , N - 1\}$.

\subsection{Sensitivity and block sensitivity}

Given a string $x \in \{0, 1\}^N$ and a set $B \subseteq \bZ_N$ (also referred to as a `block'), define $x^{B}$ as the string whose $i$th bit is $\overline{x_i}$ if $i \in B$, and $x_i$ otherwise.  In particular, let $x^i$ denote the string $x$ with its $i$th bit flipped.

For any Boolean function $f: \{0, 1\}^N \rightarrow \{0, 1\}$ and $x \in \{0, 1\}^N$, say that $B \subseteq \bZ_N$ is a \textit{sensitive block for $x$} if $f(x^{B}) \neq f(x)$.  Define $bs(f; x)$ as the largest $d$ for which there exists $d$ disjoint sensitive blocks $B_1, \ldots, B_d \subseteq \bZ_N$ for $x$.

For $b \in \{0, 1\}$, define the \textit{$b$-block sensitivity} of $f$,  or $bs_b(f)$, as $\max_{x \in f^{-1}(b)}bs(f; x)$.  Define the \textit{block sensitivity} $bs(f) = \max(bs_0(f), bs_1(f))$.   Block sensitivity was first defined by Nisan in~\cite{Nis}.

The \textit{sensitivity} of $f$, denoted $s(f)$, is defined identically to $bs(f)$, but with the further restriction that all sensitive blocks considered must be of size $1$ (thus $s(f) \leq bs(f)$).  Sensitivity, a concept which predates block sensitivity, was defined in~\cite{CDR} (and originally called `critical complexity').

\subsection{Patterns, permutations, and invariance}\label{ss2}

Define a \textit{pattern} as a string $p \in \{0, 1, *\}^N$ (note, this definition includes the ordinary strings $x \in \{0, 1\}^N$), and define the \textit{domain} of $p$, denoted $\dom(p)$, as $\dom(p) = \{i: p_i \in \{0, 1\}\}$.  We say that $p$ is \emph{defined on $i$} if $i \in \dom(p)$.  Say that two patterns $p, p'$ \textit{agree} if for all $i \in \bZ_N$, $p_i \in \{0, 1\} \Rightarrow p'_i \in \{p_i, *\}$.  Note that this condition is symmetric in $p, p'$.

For a pattern $p$  and a permutation $\sigma$ from the symmetric group $S_N$ (considered as the group of permutations on $\bZ_N$), define the \textit{$\sigma$-shift of $p$}, denoted $\sigma(p)$, by the rule $\sigma(p)_i = p_{\sigma^{-1}(i)}$.  Similarly, for subsets $B \subseteq \bZ_N$, define the $\sigma$-shifted set $\sigma(B) = \{\sigma(b): b \in B\}$.

Given a permutation group $\Gamma \leq S_N$, we say a Boolean function $f$ is \emph{invariant under $\Gamma$} if for all $x \in \{0, 1\}^N$ and $\sigma \in \Gamma, f(x) = f(\sigma(x))$.  

A permutation group $\Gamma$ is called \textit{transitive} if for all $i, j \in \bZ_N$ there exists  $\sigma \in \Gamma$ such that $\sigma(i) = j$.  An important example of a transitive permutation group is the family of cyclic shifts of the coordinates, which we'll denote by $\mathcal{T} = \{t_j$: $t_j(i) = i + j \newmod N\}_{j \in \bZ_N}$.  We say a Boolean function $f$ is \textit{transitively invariant} (or \textit{weakly symmetric}, in~\cite{Sun}) if it is invariant under some transitive group $\Gamma$.  We say $f$ is  \textit{cyclically invariant} if it is invariant under $\mathcal{T}$.

Given a pattern $p$ and a permutation group $\Gamma \leq S_N$, define the \textit{$(\Gamma, p)$-pattern matching problem} $f^{\Gamma, p}: \{0, 1\}^n \rightarrow \{0, 1\}$ by

\[ f^{\Gamma, p}(x) = 1 \Leftrightarrow \exists \sigma \in \Gamma: x \text{ agrees with } \sigma(p).      \]
Equivalently, $f^{\Gamma, p}(x) = 1 \Leftrightarrow \exists \sigma \in \Gamma$ such that $\sigma(x)$ agrees with $p$.  A function $f : \{0, 1\}^n \rightarrow \{0, 1\}$ is called \textit{minterm-transitive} if there exists a transitive group $\Gamma$ and pattern $p$ such that $f = f^{\Gamma, p}$.  $f$ is called \textit{minterm-cyclic} if in addition we may take $\Gamma = \mathcal{T}$.  Note that transitive pattern-matching functions are transitively invariant, and minterm-cyclic functions are cyclically invariant.  (Both of these subclasses were defined in~\cite{Cha}, where the terminology is explained.)

\subsection{A probabilistic inequality}\label{probineqsec}

The key tool in our construction of a minterm-transitive function with low block sensitivity is a probabilistic inequality from a paper of Janson~\cite{Jan}.  This inequality reformulates an earlier result of Suen~\cite{Suen}, which in turn generalizes another, earlier result of Janson (see~\cite{Jan} and Alon and Spencer's book \cite[Sec. 8.7]{AS} for more details).  Roughly speaking, the inequality upper-bounds the probability that a family of indicator random variables sums to zero, provided the expected value of their sum is large enough and they are `mostly independent'.  We set up and state this inequality (which will be used only in Section~\ref{uboundsec}) next.

Let $\{I_i\}_{i \in \mathcal{I}}$ be a finite family of indicator (i.e., 0/1-valued) random variables on some probability space $\Omega$.  Let $G$ be an (undirected) \textit{dependency graph} with vertex set $\mathcal{I}$ (and edges indicated by $\sim$, and with $i \nsim i$ for all $i$).  This means that, if $A, B$ are disjoint sets in $\mathcal{I}$ and $i \nsim j$ for each pair $(i, j) \in A \times B$, then the family $\{I_i\}_{i \in A}$ must be independent of the family $\{I_j\}_{j \in B}$.

For $i \in \mathcal{I}$, let $q_i := \bE[I_i]$, and let $\mu := \bE[\sum_{i \in \mathcal{I}} I_i] = \sum_{i \in \mathcal{I}}q_i$.
Let $\delta_i := \sum_{j: i \sim j}q_j$.  Let $\delta := \max_{i}\delta_i$. and $\Delta := \sum_{\{i, j\}: i \sim j} \bE[I_i I_j]$, where the sum is over unordered pairs.  Observe that $\delta$ and $\Delta$ measure in a sense the `level of dependence' among the family.  Then the promised inequality is as follows:

\begin{mythm} \label{SuenJan} \cite{Jan} $\Pr[\sum_{i \in \mathcal{I}}I_i = 0] \leq e^{-\mu + \Delta e^{2\delta}}$.
\end{mythm}

\section{Lower Bounds on $bs(f)$ for Minterm-Transitive Functions \label{lboundsec}}

In this section we prove:

\begin{mythm}\label{lbound}  If $f: \{0, 1\}^N \rightarrow \{0, 1\}$ is a nonconstant minterm-transitive function, then $bs(f) = \Omega(N^{3/7})$.  \end{mythm}

We will need the following easy observation,  essentially due to \cite{RV}:

\begin{mylem}\label{uniform} \cite{RV} If $\Gamma \subseteq S_N$ is a transitive group of permutations, $i \in \bZ_N$ is any index, and $\sigma$ is a uniformly chosen element of $\Gamma$, then $\sigma(i)$ is uniformly distributed over $\bZ_N$. \qed
\end{mylem} 

Our main new tool is the following combinatorial lemma.

\begin{mylem}\label{nicepack}  Let $B \subseteq \bZ_N$ be of size at most $N^{3/7}$, and let $\Gamma \leq S_N$ be a transitive permutation group.  If $N$ is sufficiently large, there exists a $T  \geq  \frac{1}{2}N^{3/7}$ and group elements $\Sigma = \{\sigma_1, \ldots, \sigma_T\} \subseteq \Gamma$ such that for each $i \in \bZ_N$, there are at most $3$ indices $j \leq T$ for which $i \in \sigma_j(B)$.
\end{mylem}
(Note that there is no requirement that the $\sigma_j$ all be distinct.)
\begin{proof}[Proof of Lemma~\ref{nicepack}]  Our strategy is as follows: first we select $T_0$ permutations $\sigma_j$ independently at random from $\Gamma$, where $T_0 := \lceil N^{3/7} \rceil$.  Some indices $i$ will be contained in $4$ or more of the shifted sets $\sigma_j(B)$, but we argue that with nonzero probability, we can discard at most $\frac{1}{2} N^{3/7}$ of the permutations in our collection to `fix' every such index $i$.

So let $\sigma_1, \ldots, \sigma_{T_0}$ be independent and uniform from $\Gamma$.  For each $i \in \bZ_N$, say $i$ is `bad' if $i \in \sigma_j(B)$ for at least 4 trials $j \leq T_0$.  We upper-bound the probability that $i$ is bad.  First, for any fixed trial, Lemma~\ref{uniform} tells us that $\Pr[i \in \sigma_j(B)] = |B|/N$.  Independence of the trials implies that for any fixed $4$-tuple of distinct trials $(j_1, j_2, j_3, j_4) \in [T_0]$, the probability that $i$ is in the shifted set on each trial is $(|B|/N)^4$.  Then by a union bound,
$$
\Pr[i \text{ is bad}] \leq {T_0 \choose 4}\frac{|B|^4}{N^4} < \frac{(T_0 |B|)^4}{24 N^4}
\leq \frac{((N^{3/7} + 1) N^{3/7})^4}{24 N^4} < \frac{N^{-4/7}}{23},
$$
the last step holding if $N$ is sufficiently large.  Summing over all $i \in \bZ_N$, the \textit{expected} number of bad $i$'s is less than $N^{3/7}/23$.  By Markov's bound, the probability that there are $N^{3/7}/12$ bad indices is less than $1/2$.

Now say that $i \in  \bZ_N$ is `terrible' if $i \in \sigma_j(B)$ for at least $7$ indices $j \leq T_0$.  By reasoning similar to the above, the expected number of terrible indices is at most 
$$
N{T_0 \choose 7} \left(\frac{|B|}{N}  \right)^7 < N \cdot \frac{(N^{-1/7})^7}{7!  - 1} < 1/2
$$
for sufficiently large $N$.  So the probability that \emph{any} terrible index appears is less than 1/2, and we find that with positive probability there are no terrible indices and fewer than $N^{3/7}/12$ bad indices.

We take any such outcome (specified by a sequence $\sigma_1, \ldots, \sigma_{T_0}$) and for each bad index $i$, delete from the collection all permutations $\sigma_j$ such that $i \in \sigma_j(B)$.  Since there are no terrible indices, each such deletion removes at most $6$ permutations, so the total number of permutations deleted is less than $6\cdot \left(N^{3/7}\right / 12) =  N^{3/7}/2$. The remaining collection has size greater than $N^{3/7}/2$ and satisfies the Lemma's conclusion.
\end{proof}

\begin{proof}[Proof of Theorem~\ref{lbound}]  Take any nonconstant minterm-transitive function $f = f^{\Gamma, p}: \{0, 1\}^N \rightarrow \{0, 1\}$, where $\Gamma$ is a transitive group and $p$ a pattern.  Let $B = \{i: p_i \in \{0, 1\}\}$.  Without loss of generality we may assume that $p$ contains at least $|B|/2$ 1's.  Then if we let $x \in \{0, 1\}^N$ agree with $p$ and equal 0 where $p$ is undefined, we see that $f(x) = 1$, while $f(x^{i}) = 0$ for any $i$ such that $p_i = 1$.  

Thus $bs(f) \geq bs(f; x) \geq |B|/2$.  If $|B| > N^{3/7}$, then $bs(f) > N^{3/7}/2$.  Let us assume now that $|B| \leq N^{3/7}$.  In this case, Lemma~\ref{nicepack} applies to $B$: there exist group elements $\Sigma = \{\sigma_1, \ldots, \sigma_T\} \subseteq \Gamma$ (with $T \geq \frac{1}{2}N^{3/7}$) satisfying the Lemma's conclusions.  Let $\Sigma(p) = \{\sigma_j(p): \sigma_j \in \Sigma\}$ denote our distinguished set of shifted patterns, and let $\mathcal{B}_{\Sigma} = \{B_j = \sigma_j(B)\}$ denote the corresponding collection of domains.

Consider the set $U \subseteq \bZ_N$ of indices $i$ appearing in at least two sets $B_j \in \mathcal{B}_{\Sigma}$.  At most three patterns $\sigma_j(p) \in \Sigma(p)$ from our collection are defined on any given index, so for each $i \in U$ we can select a value $v_i \in \{0, 1\}$ such that \textit{at most} one $\sigma_j(p)$ that is defined on $i$ disagrees with the setting $v_i$ there.  

Let us do the following:
\begin{itemize} 
\item Initialize $x \in \{0, 1\}^N$ to any value such that $f(x) = 0$.
\item If there exists some $i \in U$ such that $x_i \neq v_i$, and such that $f(x^{i}) = 0$, pick such an $i$ arbitrarily and set $x \leftarrow x^{i}$; otherwise halt.
\item Repeat the previous step until we halt.
\end{itemize} 
Note that $f(x) = 0$ for every value of $x$ during the algorithm's run.  Note also that the algorithm  must halt, since each step reduces the number of disagreements between $x$ and the $v_i$'s.  Now we ask the following question: looking at the final value of $x$ when the algorithm halts, how many $i \in U$ are such that $x_i$ still disagrees with $v_i$?  Call these indices `stubborn'.

First, suppose there are at least $N^{3/7}/12$ stubborn indices.  Since we halted, it must be the case that $f(x^{i}) = 1 \neq f(x)$ for each such stubborn index $i$, and thus $bs(f) \geq bs(f; x) \geq  N^{3/7}/12$.

On the other hand, suppose there are fewer than $N^{3/7}/12$ stubborn indices.  As each index $i \in \bZ_N$ appears in at most 3 sets from $\mathcal{B}_{\Sigma}$, fewer than $N^{3/7}/4$ patterns from $\Sigma(p)$ contain \textit{any} stubborn index.  If $B_j \in \mathcal{B}_{\Sigma}$ contains no stubborn indices, call it `stubborn-free'; so, there are more than $T - N^{3/7}/4 \geq N^{3/7}/4$ stubborn-free sets $B_j$.

For each $B_j \in \mathcal{B}_{\Sigma}$ define the `disagreement set' $D_j = \{i: \sigma_j(p)_i \in \{0, 1\} \wedge x_i \neq \sigma_j(p)_i\} \subseteq B_j$.  Each $D_j$ is nonempty, since $f(x) = 0$ and $f = f^{\Gamma, p}$.  Also, $f(x^{D_j}) = 1 \neq f(x)$.

Observe that if $B_j$ is stubborn-free, and $i \in D_j$, then $\sigma_j(p)$ is the \textit{only} pattern in $\Sigma(p)$ that disagrees with $x$ at $i$. Thus if $B_j, B_{j'}$ are stubborn-free, $D_j \cap D_{j'} = \emptyset$, so $bs(f; x)$ is at least the number of stubborn-free sets $B_j \in \mathcal{B}_{\Sigma}$, which we've seen is at least $N^{3/7}/4$.

Combining all of our cases, we find that $bs(f) = \Omega(N^{3/7})$. 
\end{proof}

\section{Improved Upper-Bound Example for $bs(f)$ \label{uboundsec}}

Sun~\cite{Sun} gave an example of  a minterm-cyclic function with block sensitivity $O(N^{3/7}\ln N)$.  This was the lowest block sensitivity known for any non-constant transitively invariant function.
In this section we prove the following result, improving on Sun's example:

\begin{mythm}\label{ubound}  There exist a family of nonconstant, minterm-transitive (in fact minterm-cyclic) functions $f_N: \{0, 1\}^N \rightarrow \{0, 1\}$, such that $bs(f_N) = O(N^{3/7}\ln^{1/7} N)$.  \end{mythm}

Most steps of our proof follow the outline of Sun's, but for completeness we give a self-contained proof.
Before defining the $p$ we will use to define $f_N = f^{\mathcal{T}, p}$, we give two lemmas (both from \cite{Sun}) for upper-bounding the block sensitivity of such functions.

\begin{mylem}\label{1bound}\cite{Sun} For any $f = f^{\mathcal{T}, p},  bs_1(f)\leq |\dom (p)|$.
\end{mylem}

\begin{proof}

Note if $f(x) = 1$ then some shift $t_{j_0}(p)$ of $p$ agrees with $x$.  So any collection of disjoint blocks $\{B_j\}$ for which $f(x^{B_j}) =0$ must assign a distinct coordinate from $\dom (t_{j_0}(p))$ to each $B_j$, and $|\dom(t_{j_0}(p))| = |\dom(p)|$.  This proves the Lemma. \end{proof}

Obtaining an upper bound on $bs_0(f)$ takes a bit more work.  We give some preparatory definitions.
By a \textit{4-set} in $\bZ_N$ we mean a subset of $\bZ_N$ of size 4.  If $A = \{a_1, \ldots, a_4\}$ is a 4-set, say that pattern $p$ \textit{contains a balanced shifted copy of $A$} if there exists a cyclic shift $t_j$ such that the shifted pattern $t_j(p)$ satisfies $\dom(t_j(p)) \supseteq A$, and $t_j(p)$ equals 0 on some two of the coordinates in $A$ and equals 1 on the other two.

\begin{mylem}\label{0bound}\cite{Sun} For any $f = f^{\mathcal{T}, p}$,  if $bs_0(f) \geq d$ then there exists a set $S \subseteq \bZ_N$ of size $d$, such that there is no 4-set $A \subseteq S$ for which $p$ contains a balanced shifted copy of $A$.
\end{mylem}

\begin{proof}  Say $bs_0(f) \geq d$; then there exists an $x$ and $d$ disjoint subsets $B_1, \ldots, B_d \subseteq \bZ_N$ such that $f(x^{B_k}) = 1 \neq f(x)$, for $1 \leq k \leq d$.  Thus for each $k$ there exists $j(k) \in \bZ_N$ such that $x^{B_k}$ agrees with $t_{j(k)}(p)$.  If $k \neq k'$ yet $j(k) = j(k')$, then both of $B_k, B_{k'}$ would contain each of the (nonempty set of) coordinates on which $t_{j(k)}(p)$ disagrees with $x$, contradicting disjointness; so the indices $j(1), \ldots, j(d)$ are all distinct.  Let $J = \{j(k)\}$ denote this index-set ($|J| = d$).

Let $S:= -J = \{-j: j \in J\}$ (all arithmetic in this section is mod $N$).  We claim that if $A$ is any 4-set contained in $S$, then $p$ contains no balanced shifted copy of $A$.  For suppose it did; that is, suppose there exists some $j^* \in \bZ_N$ such that distinct indices $-j(k_1), -j(k_2), -j(k_3), -j(k_4)$ are in the domain of $t_{j^*}(p)$, and that (renaming the $k_i$'s if necessary) $t_{j^*}(p)_{-j(k_1)} = t_{j^*}(p)_{-j(k_2)} = 0$ while $t_{j^*}(p)_{-j(k_3)} = t_{j^*}(p)_{-j(k_4)} = 1$.  Equivalently, $p_{-j(k_1) - j^*} = p_{-j(k_2) - j^*} = 0$ and $p_{-j(k_3) - j^*} = p_{-j(k_4) - j^*} = 1$.   

Now for $i \in [4]$, recall that $x^{B_{k_i}}$ agrees with $t_{j(k_i)}(p)$.  In particular, $x^{B_{k_i}}_{-j^*} = t_{j(k_i)}(p)_{-j^*}$, i.e.,
$$
x^{B_{k_i}}_{-j^*} = p_{-j(k_i) - j^*}.
$$  
We have seen that for $i = 1, 2$ the right-hand side equals 0, and for $i = 3, 4$ the right-hand side equals 1.  Thus the index $-j^*$ must be contained in exactly \textit{two} of the sets $B_{k_1}, \ldots , B_{k_4}$, contradicting the fact that they are disjoint.  Thus $p$ contains no balanced shifted copy of any 4-set $A \subseteq S$.   \end{proof}

We can now explain our strategy (following \cite{Sun}) to prove Theorem~\ref{ubound}: we build a pattern $p \in \{0, 1, *\}^N$ with `small' domain, so that $bs_1(f^{\mathcal{T}, p})$ is small by Lemma~\ref{1bound}.  We choose $p$ such that for \textit{any} `sufficiently large' $S \subseteq \bZ_N$, $p$ contains a balanced shifted copy of some 4-set $A \subseteq S$; this will bound $bs_0(f^{\mathcal{T}, p})$ by Lemma~\ref{0bound}.

Our pattern $p$ will have all of its 0/1 entries on $\{0, 1, \ldots, 2K - 2\}$, where $K = K_N < N/2$ is a parameter.  In this we are following~\cite{Sun}, but with some further optimization in our setting of $K$.  The key properties we need in $p$ are provided by the following Lemma:

\begin{mylem}\label{coolset} For sufficiently large $K$, there is a pattern $p$ with $\dom(p) \subseteq \{0, 1, \ldots, 2K - 2\}$ which contains a shifted balanced copy of every 4-set $A \subseteq \{0, 1, \ldots, K - 1\}$, and such that $|\dom(p)| \leq 3 K^{3/4}\ln^{1/4} K$.
\end{mylem}

Note that the `sufficiently large' requirement in Lemma~\ref{coolset} is independent of $N$.
This Lemma resembles ~\cite[Lemma 2]{Sun}, but uses a different construction and improves its parameters.  Sun defined a pattern $p$ by randomly assigning 0/1 values to a collection of translates of an explicit set; by contrast, we use a fully probabilistic construction.
We defer the proof of Lemma~\ref{coolset}.

\begin{proof}[Proof of Theorem~\ref{ubound}]  Set $K := \lceil \frac{N^{4/7}}{\ln^{1/7} N} \rceil$.  Fix a $p$ as guaranteed by Lemma~\ref{coolset} (for each sufficiently large $N$).  Let $f^{\mathcal{T}, p}$ be the corresponding minterm-cyclic pattern-matching problem.  First, by Lemma \ref{1bound}, 
$$
bs_1(f^{\mathcal{T}, p}) \leq 3 K^{3/4}(\ln K)^{1/4}
$$
$$
 = O\left(    \frac{N^{\frac{4}{7} \cdot \frac{3}{4}}}{ (\ln N)^{\frac{1}{7} \cdot \frac{3}{4}} } \cdot     (\ln N)^{1/4}   \right)       =  O\left(N^{3/7} \ln^{1/7} N \right),
$$
since $\frac{1}{4} - \frac{3}{28} = \frac{1}{7}$.

To upper-bound $bs_0(f^{\mathcal{T}, p})$, let $S \subseteq \bZ_N$ be any set of size at least $4N^{3/7}\ln^{1/7} N \geq \frac{4N}{K}$.  Following~\cite{Sun}, if we pick an interval $[a, a + K - 1]$ (mod $N$) by choosing $a \in \bZ_N$ uniformly, the expected number of elements of $S$ in the interval is at least $K\cdot \frac{4N/K}{N}  = 4$; therefore there exists some such interval which contains at least 4 elements of $S$.  Let $A \subseteq S$ be these 4 elements; since $A$ lie in an interval of length $K$, Lemma \ref{coolset} tells us that $p$ contains a balanced shifted copy of $A$.  

As $S$ was an arbitrary set of size $\geq 4N^{3/7}\ln^{1/7} N$, it follows from Lemma \ref{0bound} that $bs_0(f^{\mathcal{T}, p}) < 4N^{3/7}\ln^{1/7} N$, and hence that $bs(f^{\mathcal{T}, p}) = O(N^{3/7}\ln^{1/7} N)$.  This proves Theorem \ref{ubound}.  
\end{proof}

 \begin{proof}[Proof of Lemma~\ref{coolset}] We construct $p$ as follows: for each $0 \leq i \leq 2K - 2$, we independently set $p_i$, where for $b \in \{0, 1\}$ we have $\Pr[p_i = b] = (\frac{\ln K}{K})^{1/4}$; with the remaining probability we set $p_i = *$.


Now we prepare to apply Theorem \ref{SuenJan} from Section~\ref{probineqsec}.  Fix any 4-set $A \subseteq \{0, 1, \ldots , K - 1\}$.  For $0 \leq i \leq K - 1$, let $I_i$ be the event that $A + i$ is contained in the domain of $p$ and receives a balanced coloring by $p$ (note that $A + i \subseteq \{0, 1, \ldots 2K - 2\}$).  We will use Theorem \ref{SuenJan} to upper-bound the probability that $\sum_{i \in \mathcal{I}}I_i = 0$; then we will simply take a union bound over all possible choices of $A$.  We define $I_i \sim I_j$ to hold iff $i \neq j$ and $(A + i) \cap (A + j) \neq \emptyset$.  Note that this defines a valid dependency graph since $p$ is chosen according to a product measure.

Let us compute $\mu$ for our family of random variables.  Note that each translate $A + i$ can be given a balanced coloring by $p$ in ${4 \choose 2} = 6$ ways, and that each such coloring has probability $((\frac{\ln K}{K})^{1/4})^4 = \frac{\ln K}{ K}$.  Thus $q_i = \frac{6 \ln K}{K}$ and $\mu = 6\ln K$.

Now we bound $\delta$ and $\Delta$.  Note that each translate $A + i$ overlaps with at most 3 others, so that we clearly have $\delta = o(1)$.  Also, for each pair $A + i, A + j$ of overlapping translates, there are certainly no more than ${4 \choose 2}^2/2 = 18$ colorings of $(A + i) \cup (A + j)$ that make both translates balanced, and the probability of each one occurring is at most $((\frac{\ln K}{K})^{1/4})^5$ (since $|(A + i) \cup (A + j)| \geq 5$).  There are at $O(K)$ such pairs; thus,
$$
\Delta = \sum_{\{i, j\}: i \sim j} \bE[I_i I_j]   \leq O\left(K \cdot \frac{\ln^{5/4}K}{K^{5/4}}\right) =   o(1).
$$
Theorem~\ref{SuenJan} then tells us that $\Pr[\sum_i I_i = 0] \leq e^{-6 \ln K + o(1)} = (1 + o(1))K^{-6}$.  This is less than $K^{-4}$ for large enough $K$.

There are ${K \choose 4} < K^4/24$ 4-sets $A \subseteq \{0, 1, \ldots, K - 1\}$, so for large enough $K$, the probability that $p$ fails to contain a balanced shifted copy of \textit{any} such $A$ is, by a union bound, at most $1/24$.  Also, the expected domain size of $p$ is $2K \cdot (\frac{\ln K}{K})^{1/4} = 2K^{3/4} \ln^{1/4} K$.  Using Markov's inequality, the probability that $|\dom(p)| > 3K^{3/4} \ln^{1/4} K$ is less than $2/3$.  By a union bound we find that with nonzero probability, $p$ contains a balanced shifted copy of each 4-set $A \subseteq \{0, 1, \ldots, K - 1\}$, and simultaneously $|\dom(p)| \leq 3K^{3/4} \ln^{1/4}K$.  This proves Lemma~\ref{coolset} (and completes the proof of Theorem \ref{ubound}). 
\end{proof}

\section{Open Problems}

It seems natural to wonder if the parameters in Lemma~\ref{coolset} can be improved further to remove the log factor entirely.  (If so, we suspect a non-probabilistic approach is needed.)  This would yield a tight bound of $\Theta(N^{3/7})$ for the minimum achievable block sensitivity for nonconstant minterm-transitive functions.  

More broadly, we still hope for a better understanding of the sensitivity and block sensitivity of general transitively invariant functions.  The premier open problem in this area is whether for such functions $s(f) = N^{\Omega(1)}$; it is unsolved even for the special case of cyclically invariant functions.

\end{document}